\documentclass{article}  
\usepackage[utf8]{inputenc}
\usepackage[a4paper]{geometry}
\usepackage{hyperref}
\hypersetup{
    colorlinks=true,
    linkcolor=magenta,
    filecolor=magenta,      
    urlcolor=blue,
}

\usepackage{dsfont}
\usepackage{amssymb,amsmath,amsopn,amscd,amsfonts,eucal,amsthm}

\theoremstyle{plain}
\newtheorem{theorem}{Theorem}[section]
\newtheorem{lemma}[theorem]{Lemma}

\newtheorem{proposition}[theorem]{Proposition}

\theoremstyle{definition}
\newtheorem{definition}[theorem]{Definition}

\theoremstyle{remark}
\newtheorem{remark}{Remark}

\usepackage{graphicx,epstopdf,float,subcaption,wrapfig}
\graphicspath{{img/}{../img/}}
\usepackage{epstopdf}

\usepackage{authblk}

\def\({\left(}
\def\){\right)}
\def\One{{\mathds {1}}}
\def\la{\lambda}
\def\bydef{:=}
\DeclareMathOperator{\MP}{MP}
\def\MPP{{\operatorname{MPP}}}

\def\si{\sigma}
\def\iy{\infty}
\def\E{\ensuremath{\mathbb {E}\,}}
\def\d{\,\mathrm{d}}
\def\Pp{\ensuremath{\mathbb {P}}}
\def\t{\tau}
\def\Pois{{\operatorname{Pois}}}
\def\vs{\, \vert\,}
\newcommand{\dimensionofmpp}{d}
\newcommand{\numdim}{\dimensionofmpp}
\def\D#1{\Delta {#1}}
\def\vk{\textbf{k}}
\def\vx{\textbf{x}}
\def\vp{\textbf{p}}
\def\vz{\textbf{z}}

\def\norm#1{{\Vert {#1}\Vert}}
\def\nvk{\norm{\vk}}
\def\R{\ensuremath{\mathbb {R}}}

\def\Z{\ensuremath{\mathbb {Z}}}

\def\cov{\operatorname{Cov}}
\def\corr{ \operatorname{Corr}}
\def\extr{\operatorname{extremize}}

\def\eqod{{\stackrel{d}{=}}}
\newcommand{\PP}{{\mathbb{P}}}
\def\abs#1{\vert\, #1\,\vert}
\def\Vs{\,\Big\vert\,}

\usepackage[backend=biber,style=numeric]{biblatex}
\usepackage[linesnumbered,ruled]{algorithm2e}
\SetCommentSty{mycommfont}
\SetKwInput{KwInput}{Input}                
\SetKwInput{KwOutput}{Output}              
\SetKwInput{KwRequires}{Requires}              

\begin{document}

\title{Backward Simulation of Multivariate Mixed Poisson Processes}

\author[1]{Michael Chiu\footnote{chiu@cs.toronto.edu}}
\affil[1]{Department of Computer Science, University of Toronto}
\author[1]{Kenneth R. Jackson\footnote{krj@cs.toronto.edu}}
\author[2]{Alexander Kreinin\footnote{alex.kreinin@sscinc.com}}
\affil[2]{FE, Research and Validation, SS\&C Technologies}



\maketitle

%

\begin{abstract}
The simulation of correlated multivariate Poisson processes with negative correlation between their components has many important applications in Finance, Insurance, Geophysics, and many other areas of applied probability. Introduced in our earlier work, the Backward Simulation (BS) approach to the simulation of correlated multivariate Poisson processes is able to capture a wide range of correlation, including extreme positive and extreme negative correlation, that is not possible with other approaches such as the forward simulation approach. Moreover, the BS approach enables simple and efficient generation of sample paths of correlated multivariate Poisson processes. In this work, we extend the BS approach to multivariate mixed Poisson processes.
\end{abstract}
%
%
\section{Introduction}
\label{sec:introduction}

The simulation of dependent Poisson processes is an important problem having many applications in Insurance, Finance, Geophysics and many other areas of applied probability---see \cite{Aue,baeKrenin,BN1,BNO,Bock2,Chav,EmbPuc,Panj,Pet,Shev} and references therein. For example, in Operational risk, estimating the losses resulting from operational events requires the simulation of multivariate Poisson processes that must be calibrated to historical correlation matrices of operational events; see \cite{KAY}. The Poisson and Negative Binomial processes are some of the most popular underlying models amongst practitioners for describing the operational losses of the business units of a financial organization and the moments of claim arrivals in the insurance industry \cite{McNeil,Nesl}. Dependence between Poisson processes can be achieved by various operations applied to independent processes. One of the most popular approaches, often considered in actuarial modeling, is the Common Shock Model (CSM) \cite{McNeil,Dian,Shev}, where a third Poisson process is used to couple two independent processes. For example, let $(\nu_t^{(1)},\nu_t^{(2)},\nu_t^{(3)})$ be three independent Poisson processes with intensities $(\lambda_1,\lambda_2,\lambda_3)$, each defined as 
\begin{equation*}
	\nu_t^{(j)} = \sum^\infty_{i=0} \One (T_i^{(j)} \leq t)
\end{equation*}
for $j=1,2,3$ where the $T_i^{(j)}$, exponentially distributed random variables, are the arrival moments corresponding to the $j^{\mathrm{{th}}}$ process. Through superposition, we can obtain two correlated Poisson processes $N_t^{\,(1)} = \nu_t^{(1)} + \nu_t^{(2)}$ and  $N_t^{\,(2)} = \nu_t^{(2)} + \nu_t^{(3)}$, where
\begin{equation*}
N_t^{\,	(1)} = \big( \nu_t^{(1)} + \nu_t^{(2)} \big) := \sum^2_{j=1}\sum^\infty_{i=0} \One (T_i^{(j)} \leq t) 
\end{equation*}
and 
\begin{equation*}
N_t^{\,	(2)} = \big( \nu_t^{(2)} + \nu_t^{(3)} \big) := \sum^3_{j=2}\sum^\infty_{i=0} \One (T_i^{(j)} \leq t).
\end{equation*}

The correlation coefficient between the Poisson processes $N_t^{\,(1)}$ and $N_t^{\,(2)}$, 
having intensities $\la = \la_1 + \la_2$ and $\mu = \la_2 + \la_3$, in the CSM satisfy
$$ \rho = \frac{ \lambda_2}{\sqrt{\la \cdot \mu} }.$$
The latter relation immediately implies
$$ 0\le \rho \le \sqrt{\frac{\min(\la, \mu) }{\max(\la, \mu) }}. $$
It is clear that, in such a model, negative correlations cannot be obtained and that correlations are constant in time. The extreme correlation problem was considered in \cite{Griff} and \cite{Nels} where an optimization problem for the joint distribution was solved numerically. The problem was reduced to that of random vectors having specified marginal distributions in \cite{Kre}, where the Extreme Joint Distribution (EJD) method, a pure probabilistic, efficient, and rather simple algorithm to find the joint distributions with extreme correlations applicable to \textit{any} discrete marginal probability distribution was proposed. Connections with some classical results obtained in \cite{Frechet}, \cite{Hoeffd}, and \cite{Whitt} were also discussed. The Backward Simulation (BS) method, in conjunction with the EJD method, was developed in order to address the restrictions in the correlation structure of multivariate Poisson processes constructed using classical approaches such as the CSM. The Backward Simulation approach, considered in \cite{KAY}, allows for a wider range of correlations, both positive and negative, to be attained in comparison to the CSM. Moreover, BS allows for a dynamic correlation structure versus static correlation in the CSM; specifically, it is a linear function, in time, of the terminal correlation, $\rho(T)$, at the end of the simulation interval $[0,T]$ \cite{Kre}.

There are two general approaches to the simulation of multivariate Poisson processes---Forward and Backward simulation. The Forward approach consists of generating exponentially distributed inter-arrival times until the simulation time is at or past the simulation interval $[0,T]$. Our Backward approach is based on exploiting the conditional uniformity of Poisson processes---we first construct a joint distribution satisfying the marginal distributions with the desired correlation structure at the terminal simulation time $T$ and then generate the corresponding number of arrival moments using the conditional uniformity of the arrival times\footnote{This is also known as the \textit{order statistic property} \cite{crump1975point}}. This is one of the major advantages of the Backward approach---only the ability to sample from a suitable joint distribution at the terminal time is required; the arrival moments of the multivariate Poisson process can be generated uniformly in a coordinate-wise manner. The BS approach was extended in \cite{Kre} to the class of bivariate processes containing both Poisson and Wiener components. It also led to the introduction of the Forward Continuation (FC) of Backward Simulation, a method for extending the process simulated by BS to subsequent intervals $[nT,(n+1)T]$, where $n$ is some integer, that preserves the joint distribution at various grid points $nT$ \cite{akm}. 

The EJD method enables the construction of joint distributions that exhibit extreme dependence between the components; in other words, the EJD method constructs extreme joint distributions that extremize $\rho(T)$. Extreme joint distributions are used to generate extreme admissible correlations, from which all correlations within the admissible range can be obtained. It was extended to the multivariate setting in \cite{akm}. 

%
%


%
%
\bigskip
In this paper, we extend the Backward Simulation approach for Poisson processes to the class of Mixed Poisson processes (MPPs), which are a natural generalization of the class of Poisson processes that can be represented as a Poisson process with a random intensity \cite{Grand}. Our contribution is a method of simulation for multivariate mixed Poisson processes such that 
\begin{enumerate}
    \item any desired correlation that is admissible at the terminal simulation time can be matched,
    \item the correlation structure is a function of time.
\end{enumerate}
 Moreover, we describe the time structure of correlations for this class of processes and analyze the Forward Continuation of the Backward Simulation in the finite interval. This approach allows us to extend the model to arbitrary times.

\bigskip
The plan of this paper is as follows. In Section~\ref{sec:MPP}, we briefly discuss the basics of MPPs. Section~\ref{sec:BS} extends the BS method for Poisson processes to multivariate MPPs, allowing us to fill in our process back to time 0. The extension to MPPs is first discussed in the bivariate setting. Section~\ref{sec:ejd} briefly reviews the EJD method, necessary for the construction of joint distributions needed at the terminal time. We also discuss in Section \ref{sec:ejd} how to sample from Extreme Joint Distributions. In Section~\ref{sec:fb_mpp}, we extend the FC approach for Poisson processes to MPPs. This allows us to propagate the process forward in time to some possibly infinite horizon. 
Finally, we make some concluding remarks in Section~\ref{sec_cr}.

%
%

\section{Mixed Poisson Process}
\label{sec:MPP}

We begin by reviewing some properties of MPPs. The main results of the theory of MPPs can be found in \cite{Grand}. Recent results on the characterization of the multivariate MPP are in \cite{Zocher}. We consider a counting process
\begin{equation}
X_t = \sum_{i=1}^\iy \One\, (T_i\le t)
\end{equation}
with arrival moments $0 < T_1 < \dots < T_i < \cdots$, where $\One( \cdot )$ is the indicator function. Also, for convenience, we let $T_0 = 0$. The classical Poisson process is defined as a process with independent increments such that the 
inter-arrival times between the events $\D{T_i} \bydef T_i - T_{i-1}$ form a sequence of independent 
identically distributed random variables having an exponential distribution with parameter $\la$. 
It is well known that the number of events of $X_t$ in the interval $[0, t]$ has the Poisson distribution with 
parameter $\la t$:

\begin{equation}
\Pp(X_t=k) = e^{-\la t}\frac{(\la t)^k}{k!}, \quad k=0,1,2,\dots \quad t>0.
\label{eq:poiss_distr}
\end{equation}

A natural generalization of the Poisson distribution is to randomize the intensity parameter $\lambda$, leading to the Mixed Poisson Distribution (MPD). 

\begin{definition}[Mixed Poisson Distribution \cite{Grand}]
A discrete random variable $X$ is said to be mixed Poisson distributed, MP(U), with structure distribution U, if
\begin{align}
p_k := & \,\,\, \PP(X = k) = \E \big[ \frac{(\Lambda)^k}{k!} e^{-\Lambda} \big] \nonumber \\
= & \,\, \int_{0}^\infty \frac{(\lambda)^k}{k!} \, e^{-\lambda} dU(\lambda), \quad k=0,1,2,\dots
\label{eq:def_mpd}
\end{align}
where $\Lambda$ is a random variable distributed according to $U$.
\end{definition}

The structure distribution $U$ can be viewed as a prior distribution, which allows us to view (\ref{eq:poiss_distr}) as a \textit{conditional distribution}, given a realization of the intensity parameter $\Lambda = \lambda$ and (\ref{eq:def_mpd}) as an unconditional distribution. Another interpretation of (\ref{eq:def_mpd}) is that it is a mixture of Poisson distributions.

\begin{definition}[Mixed Poisson Process]
$X_t$ is a MPP if it is $\MP(U)$-distributed for all $t \geq 0$. The MPP is a Poisson process with a non-negative random intensity. 
\end{definition}

Lundberg \cite{Lundb} also showed that there exists a MPP for each structure distribution $U$ and that the process is uniquely defined. 

In what follows, we denote by $\MPP(U)$, the class of MPPs with structure distribution $U$. It is not difficult to see that if $X_t \in \MPP(U)$, then the probability generating function takes the form
\begin{equation}
G(t; z)\bydef \E[z^{X_t}]=\int_0^\iy e^{x t(z-1)}\d U(x)
\label{eq:mmp_moment_generating_fn}
\end{equation}
and
$$\E[X_t]=\bar\la t, \quad \si^2(X_t)=\bar\la t + \si^2(\la) t^2 $$
where
$$\bar\la = \E[\lambda] = \int_0^\iy x \d U(x), \quad \si^2(\la) = \int_0^\iy (x-\bar\la)^2 \d U(x). $$

%
%

\subsection{Conditional distribution of arrival moments}
 
It is well known that the intervals $\Delta T_i = T_i - T_{i-1}$ of a Poisson process with intensity $\la$ form a sequence of independent, exponentially distributed 
random variables:
$$\Pp (\Delta T_i \le t) = 1 - e^{-\la t}, \quad t\ge 0; \quad i=1, 2, \dots $$
This forms the basis of the forward approach to the simulation of Poisson processes. Given $n$ events to be generated and a positive intensity $\lambda$, one can sequentially generate exponentially distributed intervals, $\Delta T_i$, and determine the arrival moments,
$$ T_n=\sum_{i=1}^n \Delta T_i. $$
However, there is an alternative approach based on the fundamental property of the conditional distribution of the arrival moments \cite{RCont}. Let ${\cal T}=\{T_1$, $T_2,\dots, T_n\}$ be a sequence of $n$ independent random variables having a uniform distribution in the interval $[0, T]$:
$$
\Pp (T_i \le t ) = \frac{t}{T}, \quad 0 \le t\le T. 
$$
Denote by $\t_k$ the $k$th order statistic of ${\cal T}$, $(k=1, 2, \dots, n)$:
\begin{equation}
\label{eqn:order_statistic}
\t_1=\min_{1\le k\le n} T_k, \t_2=\min_{1\le k\le n}\{ T_k: T_k>\t_1\}, \dots, \t_n=\max_{1\le k\le n} T_k.
\end{equation}
\begin{theorem}\label{thm_ucd}
The distribution of the arrival moments of a Poisson process, $X_t$, with finite intensity in the interval $[0,T]$ conditional on the number of arrivals, $X_T = n$, coincides with the distribution of the order statistics:
\begin{equation}
\Pp (T_k \le t \vs X_T = n) = \Pp (\t_k\le t), \quad 0 \le t\le T, \quad k=1, 2,\dots, n.
\label{eq_ucd}
\end{equation}
\end{theorem}
 
The converse statement was proved in \cite{Kre}:
\begin{proposition}\label{prop_charact}
If a process, $X_t$, is represented as a random sum
$$ X_t= \sum_{k=1}^{N} \One (T_k < t) $$ 
where $\{T_k\}_{k=1}^N$ are independent, identically distributed random variables having a uniform conditional distribution,
$$ \Pp ( T_k \le t \vs N ) = tT^{-1}, \,\, k=1, 2, \dots, N$$
in the interval $[0, T]$ and the random variable $N \sim\Pois(\la T)$, then $X_t$ is a Poisson process 
with intensity $\la$ in the interval $[0, T]$.
\end{proposition}
This result leads to the BS algorithm for the multivariate Poisson processes considered 
in \cite{KAY} and \cite{Kre}. In Section \ref{sec:BS}, we generalize Proposition \ref{prop_charact} for the class of MPPs, which can be obtained by a similar construction using the random variable $N \sim \MP(U)$. 

%
%
\subsection{Negative Binomial Process}
Let us now consider the Negative Binomial (NB) process, which is a MPP with the structure distribution $U$ being the gamma distribution. The Negative Binomial process is widely used for count data that exhibit overdispersion because, unlike the Poisson process, it does not have the restriction that its mean must equal its variance. For this reason, we use the Negative Binomial distribution in numerical experiments in Section \ref{subsect:sec_C_MPP_timestruct_of_corr} and Section \ref{subsect:fwd_time_struct_corr} to compare the processes generated by BS in the mixed Poisson versus the Poisson case.

The generating function of the Negative Binomial process is given in the following Lemma, the proof of which can be found in standard texts \cite{feller2008}.
\begin{lemma} The generating function, $G(t, z)= \E[ z^{X_t} ]$, of the 
Negative Binomial Process is
\begin{equation}
G(t, z) = (\frac{b}{b+t(1-z)})^r, \quad  \abs{z}\le 1, \,\, t\ge 0, \,\, r > 0.
\label{eq_nb2}
\end{equation}
\end{lemma}
\noindent where $b$ corresponds to the probability of success and $r$ corresponds to the number of failures until the process is stopped.

\begin{remark}
Notice that our process is \textit{not} a L\'evy process---the inter-arrival times are not independent but only conditionally independent.
\end{remark}

%
%

\section{Backward Simulation of Mixed Poisson Processes}
\label{sec:BS}

Backward Simulation of Poisson processes, studied in \cite{KAY} and \cite{Kre}, relies on the conditional uniformity of the arrival moments. BS requires sampling the corresponding joint distribution, at terminal time, to obtain a vector of the number of events for each coordinate in the simulation interval $[0,T]$. The dependency structure manifests itself in the joint distribution (Section \ref{sec:ejd} discusses how correlated joint distributions can be obtained), i.e., in the sampled vector of terminal events. Each coordinate is simulated independently by drawing the corresponding number of uniform variates, which are then ordered to give the arrival moments of events.

\subsection{Fundamental Results}
Let us now show that the distribution of the arrival moments, conditional on the number of events, is also uniform for MPPs. Moreover, we show that the process generated by BS remains a MPP. First, we introduce the following two lemmas, the proofs of which can be found in \cite{Kre}. To this end, we introduce some useful notation. Let $X_t$ be a MPP and consider the points $\{T_0, T_1, \dots,T_d \}$ where $0 \leq T_0 < T_1 < \cdots < T_d \leq T$.
Denote by $\D X_i\bydef X_{T_i} - X_{T_{i-1}}, \, i=1,2, \dots, \numdim$, non-overlapping intervals.
For a $\numdim$-dimensional vector\footnote{The $d$ that we use here for the dimension of a generic vector should not be confused with the dimension of a multivariate mixed Poisson process in Section \ref{sec:ejd} },
$\ \vk=(k_1, k_2, \dots, k_\numdim)\in\Z^\numdim_+$, with non-negative integer coordinates, $k_j\ge 0$, we denote
the norm of the vector by
$$\nvk = \sum_{j=1}^\numdim k_j. $$ 
For any $\numdim$-dimensional vector, $\vx=(x_1, x_2, \dots, x_\numdim)$, with non-negative coordinates,
and $\vk\in\Z^\numdim_+$, we denote
$$ {\vx}^{\vk} \bydef \prod_{j=1}^\numdim {x_j}^{k_j}. $$

The conditional probability of the number of events in each non-overlapping interval given the number of terminal events takes the form \cite{Grand}

\begin{equation}
\mathlarger{\Pp}\bigg(\D X_1=k_1, \dots, \D X_\numdim=k_\numdim \Vs X_T= l + \sum_{j=1}^\numdim k_j \bigg) = {\vk+l \choose \vk} \cdot
{\vp}^{\,\, \vk} \cdot q^l
\label{eq:cond_prob_events_part}
\end{equation}
with the multinomial coefficient
$$ {\vk+l\choose\vk}\bydef \frac{\( l + \sum_{i=1}^\numdim\limits k_i \){!}}{ l! \cdot \prod_{i=1}^\numdim\limits k_i!} $$
$p_j=(\Delta T_j/ T)  \in [0,1]$ for $j=1,2,\dots, \numdim$, $\vp=(p_1, \dots, p_\numdim)$ and $q=1-\sum_{j=1}^\numdim p_j$.

\begin{lemma}\label{lem_TL}
Consider a discrete random variable, $\xi$, taking non-negative integer values with probabilities,  
$p_k=\Pp(\xi=k), k=0,1,2,\dots$, and denote its generating function by
$\hat p(z)=\sum_{k=0}^\iy p_k z^k$, $\abs{z}\le 1$.
Consider a sequence
\begin{equation}
q_k(x) =\sum_{m=0}^\iy p_{k+m} {k+m\choose k} x^k (1-x)^m, \quad 0\le x\le 1, \quad k = 0,1,2,\dots
\label{eq_q_k}
\end{equation}
Then, for any fixed $x \in [0,1]$, the sequence $\{q_k(x)\}$ is a probability distribution and its generating function, $\hat q(z;x)$, is
$\hat q(z;x) = \hat p(1-x+xz)$.
\end{lemma}

\begin{lemma}\label{lem_TL_2d}
Consider a discrete random variable, $\xi$, taking non-negative integer values with probabilities, $p_k=\Pp (\xi=k), k=0,1,2,\dots$,
and denote its generating function by $\hat p(z)=\sum_{k=0}^\iy p_k z^k$, $\abs{z}\le 1$. Let $\vk\in\Z^\numdim_+$ and consider the function 
$\pi: \Z^\numdim_+\to\R$ defined by, 
\begin{equation}
    \pi({\vk};\vx) = \sum_{l=0}^\iy p_{\nvk+l} {\vk+l\choose\vk} \cdot {\vx}^{\,\vk}\cdot y^l, \label{eq_q_k_m}
\end{equation}
where $\vx = (x_1,\dots,x_d)$, $x_j \geq 0$, $\sum^d_{j=1} x_j < 1$ and $y = 1 - \sum_{j=1}^\numdim x_j$. Denote by $\hat \pi(\vz; \vx)$ the generating function
$$ \hat \pi(\vz;\vx)\bydef \sum_{\vk\in\Z^\numdim_+} \pi({\vk};\vx) {\vz}^{\,\vk}, $$
where $\ \vz=(z_1, z_2, \dots, z_\numdim)$ and $\max \{ |z_1|,\dots,|z_\numdim| \} \leq 1$, then
\begin{equation}
   \hat \pi(\vz;\vx) = \hat p (1-  \sum_{j=1}^\numdim x_j (1-z_j) ). \label{eq_mgf_q_2d}
\end{equation}
\end{lemma}
With Lemma \ref{lem_TL_2d}, the vector analogue of Lemma \ref{lem_TL}, we can prove the main theorem of this section.
	
\begin{theorem}\label{thm_pb1}
    Let the process $X_t$ be represented as a random sum
\begin{equation*}
X_t = \sum^{N}_{k=1} \One (T_k< t)
\end{equation*}
where the number of random events $N \sim\MP(U)$ and $\{T_k\}_{k=1}^N$ are independent, identically distributed random variables having a uniform conditional distribution in the interval $[0,T]$, then $X_t$ is $ \MPP(U)$  in the interval $[0, T]$.
\end{theorem}

\begin{proof}

We prove the following two statements.
\begin{enumerate}
  \item At any time $t\in [0, T]$, the generating function of $X_t$ is $\E[z^{X_t}]=
  \int_0^\iy e^{xt(z-1)}\d U(x).$
  \item The increments of the process $X_t$ over disjoint intervals are conditionally independent random variables.
\end{enumerate}
The theorem follows immediately from the two results above. Let us prove the first statement. As noted in (\ref{eq:mmp_moment_generating_fn}), the  generating function of $X_T$ is
$$ \E[ z^{X_T}] = \int_0^\iy e^{xT(z-1)}\d U(x). $$
The probabilities $p_k(t) \bydef\Pp(X_t=k), \, k \in \mathbb{Z}_+, \,\, 0\le t\le T$,
satisfy
\begin{equation}
p_k(t) = \sum_{l=0}^\iy p_{k+l}(T)\, {k+l\choose k} \(\frac{t}T\)^k \cdot \(1-\frac{t}T\)^{l}, \quad k=0,1,2,\dots
\label{eq_rec_prob}
\end{equation}
In our case, $\hat p(z)=\int_0^\iy e^{xT(z-1)} \d U(x)$. Thus, the probabilities
$q_k=p_k(t) := P(X_t = k \, | \, X_T) $.
Taking $x=t T^{-1}$ in (\ref{eq_q_k}), we obtain from (\ref{eq_rec_prob}) and Lemma~\ref{lem_TL} that
$$\hat q(z)=\int_0^\iy e^{xt(z-1)} \d U(x) $$
as was to be proved.

The second statement is proved using Lemma~\ref{lem_TL_2d}. 
Let $\vx=(x_1, x_2, \dots, x_\numdim)$, satisfying the conditions listed in the statement of the lemma.

Then from (\ref{eq:cond_prob_events_part}), we have 
\begin{align}
\Pp \, (\D X_1=k_1, & \dots, \D X_\numdim=k_\numdim ) \nonumber \\ 
= & \sum_{l=0}^\iy {\vk+l \choose \vk} \cdot {\vp}^{\,\,\vk} \cdot q^l \cdot \mathbb{P}\big(\, X_T= l + \sum_{j=1}^\numdim k_j \big).
\label{eq_cjp}
\end{align}
  Now consider the generating function
$$\pi(\vz)\bydef \E\Bigl[ \prod_{j=1}^\numdim z_j^{\D X_j}\Bigr], \quad \abs{z_j}\le 1, \,\,j=1,2,\dots,\numdim. $$
Applying Lemma~\ref{lem_TL_2d} with 
$\hat p(z)\bydef \E[z^{X_T}]=\int_0^\iy e^{\la T(z-1)}\d U(\la)$, $x_j=p_j$ and $y=q$, we obtain 
$$ \pi(\vz) = \int_0^\iy \prod_{j=1}^\numdim e^{\la \t_j(z_j-1)} \d U(\la). $$
The latter relation implies that the increments of $X_t$ are conditionally independent, as was to be proved.
\end{proof}

Theorem~\ref{thm_pb1} is intuitively appealing. Indeed, $X_t$ is a Poisson process with random intensity, 
$\la$, which is determined at time $t=0$ and, therefore, measurable with respect to the 
filtration $\{\mathcal{F}_t\}_{t\ge 0}$ generated by the process $X_t$. 
The conditional distribution of the arrival moments is uniform in the interval $[0, T]$ and does not depend on the parameters of the process. 

Algorithm \ref{alg:bsmmpp} describes the Backward Simulation of multivariate MPPs (MMPPs) in detail.
A correlated multivariate mixed Poisson process, $\textbf{N}_t$, has as its marginals MPPs. 
Since the marginals are correlated and the joint distribution does not factorize, a joint distribution 
that has the desired correlation structure with the marginalized distributions satisfying the constraints 
of the given marginals is needed (Step \ref{alg1_ejd} of Algorithm \ref{alg:bsmmpp}). This is discussed 
in Section \ref{sec:ejd}. Given a vector of the counts of the number of events from the joint distribution, 
Theorem \ref{thm_pb1} applies to each marginal distribution independently.

\begin{algorithm}[H]
\label{alg:bsmmpp}
\caption{Backward Simulation of Correlated  MMPPs}
\KwRequires{Multivariate mixed Poisson distribution at terminal time $\boldsymbol{\MP(U)}=(\MP(U^{(1)}),\dots,\MP(U^{(d)}))$}
\KwOutput{Scenarios of the multivariate mixed Poisson process}
    

Generate $\textbf{N} = (N^{(1)},\dots,N^{(d)})$ where $\textbf{N}\sim \boldsymbol{\MP(U)}$ \label{alg1_ejd}

\For{each $j$} 
{ \tcp{this can be done in parallel}
Generate $N^{(j)}$ uniform random variables $\textbf{T}^{(j)} = (T_1^{(j)},\dots,T_{N^{(j)}}^{(j)})$

Sort $\textbf{T}^{(j)}$ in ascending order
}
\Return{$\textbf{T}=(\textbf{T}^{(1)},\dots,\textbf{T}^{(d)})$}
\end{algorithm}

%
%
\subsection{Time Structure of Correlations}
\label{subsect:sec_C_MPP_timestruct_of_corr}

Let us now analyze the time structure of the correlations of multivariate MPPs generated by BS. Since correlations are inherently pairwise in nature, the analysis carried out in the bivariate setting corresponds to pairs of variables in the multivariate setting. 

\begin{theorem}[Time Structure of the Correlation Coefficient]
\label{thm_pb2}
Consider a bivariate process $(X_t, Y_t)$ such that $X_t$ and $Y_t$ possess the conditional uniformity 
property. The sample paths of the processes $X_t$ and $Y_t$ are generated by BS in the interval $[0,T]$. 
Let the correlation coefficient at time $T$, $\rho(T) \bydef \corr(X_T, Y_T)$ be known. 
Then $\rho(t) = \corr(X_t, Y_t)$ takes the form 
\begin{equation}
\rho(t) = \rho(T) \cdot\frac{Z(T) }{Z(t)}, \quad 0< t\le T,
\label{eq:time_struc_corr_bs}
\end{equation}
where 
$$
Z(t)=\frac{ \si(X_t)\,\si(Y_t)}{t^2}, \quad t>0,
$$
and $\si^2(X_t)$ denotes the variance of $X_t$. 
\end{theorem}
\begin{proof}

First, we show that the generating function of the process, 
$$\hat g(t, z, w)\bydef 
\E[z^{X_t} w^{Y_t}], \quad \abs{z}\le 1, \abs{w}\le 1
$$
satisfies the equation
\begin{equation}
\hat g (t, z, w)=\hat g (T, 1 -tT^{-1} + ztT^{-1}, 1 -tT^{-1} + wtT^{-1}).
\label{eq_gf_XY}
\end{equation}
To this end, note that for $0 \leq m \leq k$ and $0 \leq n \leq l$,
\begin{align}
\Pp(X_t=m, \, & Y_t =n \vs X_T=k, Y_T=l) = \\
& {k\choose m} \(\frac{t}{T}\)^m \(1-\frac{t}{T}\)^{k-m} {l\choose n}  
\(\frac{t}{T}\)^n \(1-\frac{t}{T}\)^{l-n} \nonumber
\end{align}
since at the end of the simulation interval $T$ there are $k$ events in total for $X_T$ 
and $l$ events 
in total for $Y_T$, the probability of the number of events $m$ and $n$ by a certain 
time $t$ can be viewed 
as a Bernoulli trial with probability of success $(t/T)$. Taking expectation, we obtain
\begin{align*}
\E[ z^{X_t} w^{Y_t} \vs & (X_T = k, Y_T = l)] \\
 &=\sum_{m=0}^k\sum_{n=0}^l z^m w^n \Pp\(X_t=m, Y_t=n\vs X_T=k, Y_T=l\) \\
& =\(1-\frac{t}T + z\frac{t}T\)^k \(1-\frac{t}T + w\frac{t}T\)^l.
\end{align*}

Denote $P_{kl} = \Pp\(  X_T=k, Y_T=l\)$. Then we find
\begin{eqnarray*} 
\hat g(t, z, w)&=&\E\Bigl[ \E[ z^{X_t} w^{Y_t}\vs (X_T, Y_T)] \Bigr] \\
 &=& \sum_{k\ge 0}\sum_{l\ge 0} P_{kl}  \(1-\frac{t}T + z\frac{t}T\)^k 
 \(1-\frac{t}T + w\frac{t}T\)^l \\
&=& \hat g\(T, 1 -tT^{-1} + ztT^{-1}, 1 -tT^{-1} + wtT^{-1}\).
\end{eqnarray*}  
Equation~(\ref{eq_gf_XY}) is derived. Differentiating $ \hat g(t, z, w)$ twice, we find 
\begin{equation}
	\cov(X_t, Y_t) =\frac{t^2}{ T^{2}} \cov(X_T, Y_T) \label{eq:covariance_t_T}
\end{equation}
from which we obtain
\begin{eqnarray*}
\rho(t) &=& \frac{\cov(X_t, Y_t)}{\si(X_t)\si(Y_t)} \\
&=& \frac{t^2}{ T^{2}}\cdot \frac{ \cov(X_T, Y_T) }{\si(X_t)\si(Y_t)} \\
&=& \frac{t^2}{ T^{2}}\cdot \frac{ \cov(X_T, Y_T) }{\si(X_T)\si(Y_T)} 
                             \cdot\frac{\si(X_T)\si(Y_T)}{\si(X_t)\si(Y_t) } \\
&=& \rho(T) \frac{t^2}{ T^{2}}\cdot \frac{\si(X_T)\si(Y_T)}{\si(X_t)\si(Y_t) }  \\
&=& \rho(T)   \cdot\frac{Z(T) }{Z(t)}.
\end{eqnarray*}
Theorem~\ref{thm_pb2} is thus proved.
\end{proof}

In the Poisson case, the auxiliary function $Z(T)/Z(t)$ in (\ref{eq:time_struc_corr_bs}) reduces to
$tT^{-1}$. Thus, the correlation structure is linear in time in the simulation interval $[0,T]$. 
This is not true in general for MPPs. For example, for the Negative Binomial processes, 
the auxiliary functions take the form 
\begin{align*}
\rho(t) & = \rho(T) \cdot \frac{Z(T)}{Z(t)} \\
& = \rho(T) \cdot \frac{t}{T} \cdot \sqrt{\frac{(\bar{\lambda}_X 
+ \sigma^2(\lambda_X)T)(\bar{\lambda}_Y
+ \sigma^2(\lambda_Y)T)}{(\bar{\lambda}_X + \sigma^2(\lambda_X)t)(\bar{\lambda}_Y 
+ \sigma^2(\lambda_Y)t)}}.
\end{align*}
The graph of the correlation function is presented in Figure \ref{fig:nb_bkwkdsim}, where the good agreement of the theoretical and empirical results can be seen. 

\begin{figure}[h!]
\centering
\includegraphics[width=\linewidth]{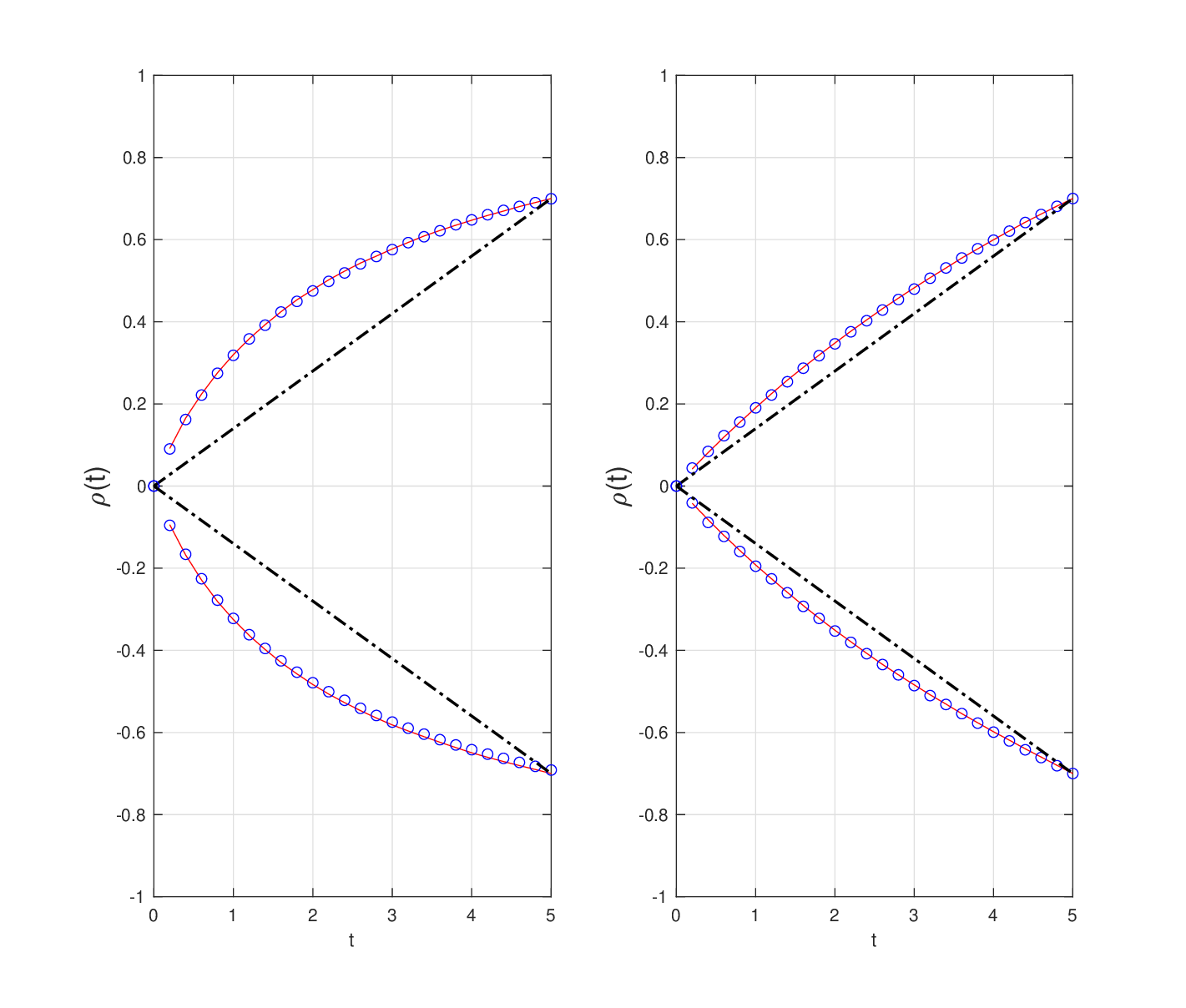}
\caption{Depicted by the red line and the blue circles are the dynamic correlation structures of two bivariate Negative Binomial (NB) processes. The red line represents the theoretical correlation structure as described in Theorem \ref{thm_pb2}. The blue circles represent the correlation structure recovered by Monte Carlo simulation of bivariate NB processes by BS. In the left figure, the first process has mean 3 and variance 1, while the second process has mean 5 and variance 30. In the right figure, the first process has mean 3 and variance 1, while the second process has mean 30 and variance 5. In both figures, the bivariate NB processes are calibrated to a positive correlation coefficient of $\rho(T)=0.7$ and a negative correlation coefficient of $\rho(T)=-0.7$. The dotted black line represents a bivariate Poisson process with mean parameters 3 and 30 simulated via BS. The bivariate Poisson processes are also calibrated to a positive correlation coefficient of $\rho(T ) = 0.7$ and a negative correlation coefficient of $\rho(T) = -0.7$.}
\label{fig:nb_bkwkdsim}
\end{figure}

%
%
\subsection{Forward vs Backward Simulation and their Correlation Boundaries}
\label{subsec:fwd_vs_bkwd_corr_bound}

Given a simulation interval $[0,T]$, stochastic processes are usually simulated forwards in time. 
This is due to the fact that it is conceptually natural and technically simpler to do so. 
However, it is not always the most suitable choice. This can be seen in the Forward Simulation (FS) 
of a correlated bivariate Poisson process $(N_t^{(1)},N_t^{(2)})$ where $N_t^{(i)} \sim \text{Poiss}(\mu_i)$ 
and $\{\Delta T_k^{(i)}\}_{k \geq 0}$ denotes the sequence of inter-arrival times for process $i\in\{1,2\}$. 
Forward Simulation of counting processes like Poisson processes consists of repeatedly simulating 
the inter-arrival times $\{\Delta T_k^{(i)}\}_{k \geq 0}$ while $\sum_k \Delta T^{(i)}_k \leq T$. 
The sequence of inter-arrival times represents a sample path of the counting process. In the Poisson case, 
the inter-arrival times are exponentially distributed, $\mathbb{P} (\Delta T_k^{(i)} \leq t) = 1 - e^{-\mu_i t}$. 
We must rely on the Fr\'{e}chet-Hoeffding Theorem\footnote{For discrete distributions, Fr\'{e}chet-Hoeffding is 
equivalent to the EJD theorem 
in 2-dimensions \cite{Kre}.} in \cite{Frechet,Hoeffd} to induce dependence between the marginal distributions 
of the inter-arrival times in the case of FS, which gives us the relations
\begin{equation}
\mu_1 \Delta T_k^{\,(1)} = \mu_2 \Delta T_k^{\,(2)} \, , \quad k = 1,2,\dots
\label{eq:frechet_hoefdding_max}
\end{equation}
to obtain extremal positive dependence between the distributions of the inter-arrival times and
\begin{equation}
\exp\,(-\mu_1 \cdot \Delta T_k^{\,(1)}) + \exp\,(-\mu_2 \cdot \Delta T_k^{\,(2)}) = 1
\label{eq:frechet_hoefdding_min}
\end{equation}
to obtain extremal negative dependence. 

We claim that the relations (\ref{eq:frechet_hoefdding_max}) and (\ref{eq:frechet_hoefdding_min}) lead to 
extreme correlations of the process under FS. In the case of extremal positive 
dependence, (\ref{eq:frechet_hoefdding_max}) implies that 
\begin{equation}
    \mu_1 T_k^{\,(1)} = \mu_2 T_k^{\,(2)} \, , \quad k = 1,2,\dots
    \label{eq:fwd_sim_increment_fh}
\end{equation}
Define $\kappa = \mu_1 / \mu_2$. Obviously,   $0 < \kappa < \infty$. 
We show that for all $t > 0,$
\begin{equation}
    N^{(1)}_t = N^{(2)}_{\kappa t}.
    \label{eq:counting_proc_FS_relation_max_dep}
\end{equation}
Suppose that $N^{(1)}_t = m$ for some $t > 0$ where $m$ is an integer. The arrival moments for $N^{(1)}_t$ satisfy the inequality
\begin{equation*}
    T_m^{(1)} \leq t < T_{m+1}^{(1)}.
\end{equation*}
It follows immediately from (\ref{eq:fwd_sim_increment_fh}) that, the arrival moments for $N^{(2)}_t$ must satisfy 
\begin{equation*}
    T^{\,(2)}_m = \kappa \, T_m^{\,(1)} \quad \text{for all \,} m=1,2,\dots
\end{equation*}
This implies that $T^{(2)}_m \leq \kappa t < T^{(2)}_{m+1}$ which in turns implies that $N^{(2)}_{\kappa t} = m$. 
Thus we have shown (\ref{eq:counting_proc_FS_relation_max_dep}) since $m$ is arbitrary. 
Now let us compute the correlation coefficient of such a process in the case $\kappa > 1$. $N^{(i)}_{\kappa t}$ can then be written as
\begin{equation*}
    N^{(i)}_{\kappa t} = N^{(i)}_t + \Delta N^{(i)}_{\kappa t}
\end{equation*}
where $i={1,2}$ and $\Delta N^{(i)}_{\kappa t}$ represents the increment of the $i$th process in the interval $[t,\kappa t]$ 
and is independent of $N_t^{(1)}$. Then we obtain
\begin{align*}
    \E [ N^{(1)}_t N^{(2)}_t ] & = \E [ N^{(2)}_{\kappa t} \cdot N^{(2)}_t  ] \\
    & = \E [ ( N^{(2)}_t )^2 ] + \E [ N^{(2)}_t \Delta N^{(2)}_{\kappa t} ]
\end{align*}
and 
\begin{equation*}
    \cov(N_t^{(1)},N_t^{(2)}) = \sigma^2 (N_t^{(2)}).
\end{equation*}
The latter relation  implies that
\begin{equation*}
    \rho( N_t^{(1)} , N_t^{(2)} ) = \frac{1}{\sqrt{\kappa}} \, , \qquad \textnormal{where} \,\, \kappa \geq 1.
\end{equation*}
Similar reasoning in the case $0 < \kappa < 1$ leads to
\begin{equation*}
    \rho( N_t^{(1)} , N_t^{(2)} ) = \sqrt{\kappa}.
\end{equation*}

This allows us to compare the correlations obtained from the FS case to the correlations obtained in the BS case. From the last two equations above, we can see that the notion of extreme dependence obtained via the Fr\'{e}chet-Hoeffding theorem results in a correlation coefficient that is a function of the intensities. This is very restrictive and precludes the possibility of calibrating to data. In contrast, the BS approach allows for the construction of processes with \textit{any} desired correlation that is within the range of admissible correlations. Further comparisons of the Forward vs the Backward approaches can be found in \cite{Chiu,Kre}.

%
%

\section{Backward Simulation and Extreme Joint Distributions}
\label{sec:ejd}

We showed in Section \ref{sec:BS} that the conditional uniformity property holds for the class of Mixed Poisson 
processes and that the process resulting from Backward Simulation with the number of events at terminal simulation 
time $T$ generated by a MPD is indeed a MPP. Backward Simulation for the class of Mixed Poisson processes relies on 
the knowledge of the joint MPD at terminal time $T$, but how do we construct a multivariate MPD with some desired 
dependency structure in the first place? In this section, we briefly review the work in \cite{akm} and \cite{Kre} in order for this paper to be self-contained. Moreover, some details are explained more clearly here than in \cite{akm} and \cite{Kre}. We address the general problem of constructing multivariate joint distributions from given marginal distributions 
such that the linear correlation coefficient between the marginals are equal to some desired correlations. We also 
discuss how to sample from such multivariate joint distributions. 

%
%
\subsection{The Bivariate Case}
We begin by describing the main ideas in 2-dimensions to build some intuition before presenting the general $d$-dimensional case. To that end, suppose we have a discrete bivariate distribution $P$ with marginals $Q^{(1)}$ and $Q^{(2)}$. Clearly, the admissible linear correlation coefficient $C$ between the marginals is bounded by some maximum attainable correlation $\hat{C}^{\,(1)}$ and some minimum attainable correlation $\hat{C}^{\,(2)}$. Moreover, \textit{every} admissible correlation $C$, can be represented as a convex combination of the extreme correlations 
\begin{equation}
	C = w \, \hat{C}^{\,(1)} + (1-w)\,\hat{C}^{\,(2)}
	\label{eq:linear_combo_corr}
\end{equation}
for some $w \in [0,1]$. The extreme correlations are clearly extreme points. Extreme Measures in the bivariate case are defined as follows
\begin{definition}[Extreme Measures in 2-dimensions]
Extreme Measures are solutions to the following infinite dimensional Linear Program (LP)

\begin{align}
    \extr & \quad h(P)  \label{eq:optimization_problem_2D} \\
    \text{subject to} & \quad \sum^\infty_{j=0} P_{ij} = Q^{(1)}_i, \quad i=0,1,\dots  \nonumber \\
    & \quad \sum^\infty_{i=0} P_{ij} = Q^{(2)}_j, \quad j=0,1,\dots \nonumber \\ 
    & \quad P_{ij} \geq 0 \quad i,j=0,1,\dots \nonumber
\end{align}
where  $\sum_{i=0}^\infty Q^{(1)}_i=\sum_{j=0}^\infty Q^{(2)}_j=1$. Extremize \textnormal{denotes either max or min and the objective function is} 
\begin{equation}
    h(P) := \mathbb{E}[X_1 X_2] = \sum^\infty_{i=0}\sum^\infty_{j=0} ij\, P_{ij}
    \label{eq:objective_fn_2D} 
\end{equation}
where $P_{ij}=\PP( X_1 = i, X_2 = j)$.
\end{definition}

For completeness, we mention that the infinite dimensional LP\footnote{In practice, probability distributions are truncated to some desired accuracy; we are really dealing with linear programs.} (\ref{eq:optimization_problem_2D}) is a Monge Kantorovich Problem (MKP). This aspect of the problem is not immediately relevant to us; we refer to standard references such as \cite{rachev1998mass} for more details.

The solution to (\ref{eq:optimization_problem_2D}) is an Extreme Joint Distribution that \textit{determines} the Extreme Measures $\hat{P}^{\,(1)}$ and $\hat{P}^{\,(2)}$ which have a one-to-one relationship to the extreme correlations $\hat{C}^{\,(1)}$ and $\hat{C}^{\,(2)}$ \cite{Chiu}. The Extreme Measures (\ref{eq:optimization_problem_2D}) lead to extreme correlations since extremizing the bivariate expectation extremizes the linear correlation coefficient as can be seen in (\ref{eq:objective_fn_2D}). 
Moreover, let
\begin{equation}
	P = w \, \hat{P}^{\,(1)} + (1-w) \, \hat{P}^{\,(2)}
	\label{eq:linear_combo_em}
\end{equation}
where $w$ is the solution of (\ref{eq:linear_combo_corr}) and $\hat{P}^{(1)}$ and $\hat{P}^{(2)}$ are the Extreme Measures having extreme correlations $\hat{C}^{(1)}$ and $\hat{C}^{(2)}$, respectively. Then, it is not hard to show that $P$ is a discrete bivariate probability distribution with marginals $Q^{(1)}$ and $Q^{(2)}$ and correlation coefficient $C$. This insight allows us to reduce the problem of calibration to a simpler problem of solving a linear equation. Thus, if extreme joint distributions can be computed, they can be used to generate extreme correlations (extreme points) to calibrate to the given correlation. If the calibration fails---there is no solution to the linear equation (\ref{eq:linear_combo_corr}) with $w \in [0,1]$---it implies that \textit{no bivariate process with the marginal distributions $Q^{(1)}$ and $Q^{(2)}$ and correlation $C$ exists}. In such an event, the assumptions of the parameter values (including the inference procedures to obtain them) and any raw data should be checked.

%
%
\subsection{The General Case}

The 2-dimensional case described in the previous subsection generalizes to the $d$-dimensional case ($d>2$) described below. However, instead of dealing with a single correlation coefficient, in the general case, we consider a $d\times d$ correlation matrix $C$ where $C_{ij}$ represents the linear correlation coefficient between marginal distributions $Q^{(i)}$ and $Q^{(j)}$. Clearly, we now have to consider more general notions of extremal dependency. One concept of extremal dependency consistent with observations of correlations matrices is to consider only pairwise extremal dependence. That is, we consider pairwise monotonicity, which represents the strongest type of association between random variables and implies maximally positive (comonotonicity) and negative values (antimonotonicity) for the linear correlation coefficient; see \cite{puccetti2015} for more details on extremal dependence concepts in multivariate settings and \cite{Kre} for details on monotonicity as it relates to distributions. In contrast to the bivariate case, there are $n = 2^{\,d-1}$ extreme correlation matrices $\hat{C}^{(j)}$, which are also extreme points \cite{akm}, each described by a \textit{monotonicity structure}.

\begin{definition}[Monotonicity Structure]
A monotonicity structure $\textbf{e}^{(j)}$, where $j \in \{1,\dots,n\}$, is a binary vector describing the pairwise extremal dependency structure between the marginal distributions
\begin{equation}
\textbf{e}^{(j)} = (e_1^{(j)},\dots,e_d^{(j)})
\label{eq:monotonicity_structure}
\end{equation} 
where
\begin{equation*}
e_i^{(j)} = 
\begin{cases}
1, \quad & \text{if $X_1$ and $X_i$ are antimonotone} \cr
0, \quad & \text{if $X_1$ and $X_i$ are comonotone}
\end{cases}
\end{equation*}
assuming that $e_1^{(j)} = 0$. If $e^{(j)}_i = e^{(j)}_k$, then marginal distributions $Q^{(i)}$ and $Q^{(k)}$ have a comonotone dependency relationship and an antimonotone dependency relationship otherwise.
\end{definition}

\begin{remark}
Note that whether $e_1^{(j)}$ is initially set to 0 or 1 does not matter and is an arbitrary choice. 
\end{remark}

Similar to the bivariate case, for each $j=1,2,\dots,n$, each extreme correlation matrix $\hat{C}^{(j)}$ is associated with an Extreme Measure $\hat{P}^{(j)}$, described by a monotonicity structure $\textbf{e}^{(j)}$, as described below.

\begin{definition}[Extreme Measures in $d$-dimensions]
Extreme Measures are solutions to the following multi-objective infinite dimensional LPs

\begin{subequations}\label{eq:multiobjective_problem}
\begin{align}
	\extr & \quad \,\, h_{k,l}^{(j)}(P) &  1 \leq k < l \leq d \label{eq:multiobjective_obj_fn} \\
	\text{subject to} & \quad \sum_{j \in \mathcal{I}_k} \sum^\infty_{i_j = 0} P^{(j)}_{i_1,...i_{k-1},i_k,i_{k+1},...,i_d} = Q^{(k)}_{ i_k } & \parbox{5.5em}{$k=1,2,\dots$ \\ $i_k$ = 0,1,\dots} \label{eq:multiobjective_problem_constraints} \\
	& \quad P_{i_1,\dots,i_d} \geq 0 &\nonumber
\end{align}
\end{subequations}
where 
\begin{equation}
\extr \,\, h_{k,l}^{(j)}(P) = 
\begin{cases}
\max h_{k,l}^{(j)}(P) \qquad \text{if} \,\,\,\, e_k^{(j)} = e_l^{(j)}  \\
\min h_{k,l}^{(j)}(P) \qquad \, \text{if} \,\,\,\, e_k^{(j)} \ne e_l^{(j)} 
\end{cases}
\nonumber
\end{equation}
$\mathcal{I}_k= \{ j : 1 \leq j \leq d, j \neq k \}$, $Q^{(k)}$ represents the $k$-th given marginal distribution and each objective function takes the form
\begin{equation}
	h_{k,l}(P) = \sum^\infty_{i_k=0}\sum^\infty_{i_l=0} i_k i_l \,
	 P^{\,(k,l)}_{i_k,i_l} \quad\quad 1 \leq k < l \leq d
	\label{eq:objective_fn_nd}
\end{equation}
where, similarly,
\begin{equation}
P^{\,(k,l)}_{i_k,i_l} = \sum_{j \in \mathcal{I}_{k,l}} \sum^\infty_{i_j = 0} P_{i_1,...i_{k-1}, i_k, i_{k+1},... ,i_{l-1}, i_l, i_{l+1},...,i_d}\nonumber
\end{equation}
with  $\mathcal{I}_{k,l}= \{ j : 1 \leq j \leq d, j \neq k, \, j \neq l \}$.
\end{definition}

\begin{remark}
There are $m=d(d-1)/2$ objective functions, where each $h_{k,l}(p)$ extremizes the dependency between a pair of coordinates.
\end{remark}

The multi-objective program (\ref{eq:multiobjective_problem}) is, in fact, a multi-objective multi-marginal MKP, the solutions of which determine Extreme Measures. Potential solutions of (\ref{eq:multiobjective_problem}) are multivariate probability measures $P$ which are tensors. Thus, the multi-objective problem is not only tedious to program but practically prohibitively expensive to compute (in terms of both time and storage) for moderate $d$ \cite{Chiu}. One approach that leads to a computable solution to these infinite dimensional problems (\ref{eq:optimization_problem_2D}) 
and (\ref{eq:multiobjective_problem}) is given by the Extreme Joint Distribution (EJD) Theorem, which gives a semi-analytic 
form describing completely the extreme joint distribution. 
\begin{theorem}[EJD Theorem in $d$-dimensions]
Given marginal cumulative distribution functions $F^{(1)}, F^{(2)}, \dots, F^{(d)}$ on $\mathbb{Z}_+$, corresponding to 
the marginal distributions $Q^{(1)}, Q^{(2)}, ..., Q^{(d)}$ in the constraints (\ref{eq:multiobjective_problem_constraints}) and a monotonicity 
structure $\textbf{e}^{(j)}$, where $j \in \{1,\dots,n\}$, the corresponding Extreme Measure is defined by the probabilities
\begin{align}
\label{eq:ejd_p_nd}
\hat{P}^{(j)}_{i_1,\dots,i_\dimensionofmpp} & = \, [\min(\bar{F_1}(i_1 - e^{(j)}_1; e^{(j)}_1),\dots,\bar{F_\dimensionofmpp}(i_\dimensionofmpp-e^{(j)}_\dimensionofmpp;e^{(j)}_\dimensionofmpp))  \\
	& \quad\quad - \max (\bar{F_1}(i_1 + (e_1^{(j)}-1); e^{(j)}_1),\dots,\bar{F_\dimensionofmpp}(i_\dimensionofmpp+(e^{(j)}_\dimensionofmpp-1);e^{(j)}_\dimensionofmpp))]^+ \nonumber
\end{align}
where $[\boldsymbol{\cdot}]^+ = \max(0,\boldsymbol{\cdot})$ and $\bar{F_k}$ is defined as
\begin{equation}
\bar{F_k}(i_k; e_k^{(j)}) = 
	\begin{cases}
		F^{(k)}(i_k)	\quad & \quad \text{if} \,\,\, e_k^{(j)}= 0 \\
        1 - F^{(k)}(i_k) \quad & \quad \text{if} \,\,\, e_k^{(j)} = 1.
	\end{cases}
\label{eq:ejd_nd_aux_fn}
\end{equation}
\label{thm:ejd_d_dim}
\end{theorem}

An accompanying algorithm, the EJD algorithm, provides an efficient numerical method to solve (\ref{eq:multiobjective_problem}) by computing the extreme joint distributions in (\ref{eq:ejd_p_nd})  and their corresponding supports; see \cite{Kre} and \cite{akm} for more details. Note that $\hat{P}^{(j)}$ is very sparse in that most $\hat{P}^{(j)} _{i_1,\dots,i_\dimensionofmpp} = 0$. A complete exposition of the details in the general case can be found in \cite{Chiu}.   

%
%

\subsection{Sampling from Multivariate Extreme Measures}

There are two attractive features of the EJD approach which make sampling from multivariate Extreme Measures simple. The first is that Extreme Measures $\hat{P}^{(k)}$ are monotone distributions \cite{Kre}. Consequently, their     support remains a graph in higher dimensions. This is very convenient for sampling as this means that Extreme Measures can be sampled from via the inverse CDF method. Second, any discrete multivariate probability measure with specified marginals and some desired dependency structure can be represented as a convex combination of Extreme Measures. That is, the one-to-one relationship between (\ref{eq:linear_combo_corr}) and (\ref{eq:linear_combo_em}) extends to the multidimensional case as follows. We first find coefficients $(w_1,\dots,w_n)$ that satisfy $w_j \geq 0$ for $j=1,2,\dots,n$, $\sum_{j=1}^n w_j = 1$ and
\begin{equation}
C = w_1 \hat{C}^{\,(1)} + \dots + w_n \hat{C}^{\,(n)}
\label{eq:calib_jdim}
\end{equation}
and then set
\begin{equation}
P = w_1 \hat{P}^{\,(1)} + \dots + w_n \hat{P}^{\,(n)}
\label{eq:calib_jdim_meas}
\end{equation}
where $\hat{P}^{(j)}$ is the extreme measure satisfying the LP (\ref{eq:multiobjective_problem}) and having the extreme correlation matrix $\hat{C}^{(j)}$.  Since $w_i \geq 0$ for $i = 1,2,...,n$ and $\sum_{i=1}^n w_i = 1$, it follows immediately that $P$ is a probability measure.  Moreover, it follows from the linearity of sums that $P$ has the correlation matrix $C$ given on the left side of (\ref{eq:calib_jdim}).  In addition, since each $\hat{P}^{(j)}$ has marginal distributions $Q^{(1)}, \dots, Q^{(d)}$, it follows that $P$ also has marginal distributions $Q^{(1)},\dots, Q^{(d)}$. 

Note that (\ref{eq:calib_jdim}) can be converted to a constrained system of linear equations by flattening each 
extreme correlation matrix $\hat{C}^{\,(j)}$ into a column vector $A_j \in \mathbb{R}^m$ where $m = d(d-1)/2$.
Since $C$ and all $\hat{C}^{\,(j)}$  are symmetric with 1s on their diagonal, this can be done by taking each row 
in the strictly upper triangular part of each $\hat{C}^{\,(j)}$, appending them into a row vector and taking the transpose
to be $A_j$ to obtain $\textbf{A} = [A_1,\dots,A_n] \in \mathbb{R}^{m \times n}$, representing the extreme points 
of our problem in correlation space. Similarly, we can flatten the correlation matrix $C$ on the left side 
of (\ref{eq:calib_jdim}) to a vector $b \in \mathbb{R}^m$. Then (\ref{eq:calib_jdim}) and the constraints $w_j \geq 0$ for $j = 1, 2, ...,n$ and $\sum_{j=1}^b w_j = 1$ are equivalent to the constrained system of linear equations
\begin{subequations}\label{simple_calibration_prob}
\begin{align}
\textbf{Aw} & = b \label{calib_prob_extrpts_constraints}  \\
\textbf{1}^T \textbf{w} & = 1 \label{calib_prob_constraints}  \\
w_j & \geq 0 \quad j = 1, 2, \dots, n. \label{calib_prob_weights_geq1}
\end{align}
\end{subequations}
There are many possible solutions to the constrained system of equations (\ref{simple_calibration_prob}). One approach is to choose a suitable objective function\footnote{This is also the subject of future work.} and then use (\ref{simple_calibration_prob}) as the constraints for an optimization problem with that objective function.  However, if the goal is just to find any solution to (\ref{simple_calibration_prob}), then a simpler approach is to reformulate (\ref{simple_calibration_prob}) as 
\begin{subequations}\label{simple_reformulate}
\begin{align}
\hat{A}w &= \hat{b} \\
w_j & \geq 0  \qquad j=1,2,\dots,n
\end{align}
\end{subequations}
where $\hat{A}$ is $A$ with the row $1^T$ appended to the bottom of it and $\hat{b}$ is $b$ with a 1 appended to the bottom of it. Then note that (\ref{simple_reformulate}) has the form of the standard constraints for a Linear Programming Problem (LPP). Moreover, the first stage of many LPP codes finds a solution to (\ref{simple_reformulate}). As explained in Section 13.5 of \cite{nocedal2006numerical}, one standard approach to finding a solution to (\ref{simple_reformulate}) is to solve the LPP 
\begin{subequations}\label{phase_1_simple_calib+prob}
\begin{align}
\min & \quad \textbf{1}^T z \\
\text{subject to} & \quad \hat{A}w + Ez = \hat{b} \label{phase_1_constraint1} \\
& \quad (w,z) \geq 0 \label{phase_1_constraint2}
\end{align}
\end{subequations}
where $z \in \mathbb{R}^{m+1}$ and $E$ is a $(m+1)\times(m+1)$ diagonal matrix such that $E_{ii} = +1$ if $\hat{b}_i \geq 0$ and $E_{ii} = -1$ if $\hat{b}_i < 0$. Clearly, $w = 0$ and $z = \abs{b}$ satisfies the constraints (\ref{phase_1_constraint1}) and (\ref{phase_1_constraint2}). So, we can use $w = 0$ and $z =\abs{b}$ as a staring point for the simplex method to solve (\ref{phase_1_simple_calib+prob}). It's clear from the constraint $z \geq 0$ that the solution satisfies $\textbf{1}^T z \geq 0$. Moreover, if $\textbf{1}^T z = 0$ then $z = 0.$ Hence, (\ref{phase_1_simple_calib+prob}) has a solution $\textbf{1}^T z = 0$ if and only if $\hat{A}w = \hat{b}$, $w \geq 0$ has a solution. Hence, the simplex method applied to (\ref{phase_1_simple_calib+prob}) will find a solution to (\ref{simple_calibration_prob}), if a solution exists. 

After obtaining $\{w_j\}^n_{j=1}$ through calibration, as described above, the decomposition of the desired measure as a convex combination of Extreme Measures (\ref{eq:calib_jdim_meas}) provides an easy method for simulation. Since $w_j \geq 0$ for $j=1,2,\dots,n$ and $\sum_{j=1}^n w_j = 1$ we can view $w_j$ as the probability that a draw of $P$ comes from the Extreme Measure $\hat{P}^{(j)}$. Moreover, as noted at the beginning of this subsection, we can easily sample from $\hat{P}^{(j)}$. Therefore, we can utilize methods of discrete random variate simulation (see \cite{Devroye} for a comprehensive exposition) to generate a random variable that has the distribution $P$.

\begin{remark}
Despite the fact that the problem size grows exponentially in $d$, due to the structure of the problem (\ref{simple_calibration_prob}) and the fact that the simplex method needs to explicitly access $m+1$ columns of $\hat{A}$ at a time (assuming you have some clever way to decide which new vector to bring into the active set at each step of the simplex method without explicitly accessing all the columns of A that are in the inactive set) the LP (\ref{phase_1_simple_calib+prob}) can be solved for a surprisingly large $d$, e.g., $d = 51$, which corresponds to $n = 2^{50} \approx 10^{15}$; see \cite{Chiu} and \cite{zoe}. 
\end{remark}

\subsection{Applications}

The novelty in our Backward Simulation approach to the construction of multivariate mixed Poisson processes is that we are able to calibrate to a given correlation structure in the form of a correlation matrix and attain extreme correlations. One application of our BS approach is to Operational Risk \cite{KAY,Kre} where it is necessary to compute regulatory capital that is allocated in the event of an operational loss. Generally, this is accomplished using Monte Carlo scenario generation, which requires calibration to historical data. The BS approach captures ranges of correlation that have been empirically observed (such as negative correlations ) but could not be captured before. The BS approach has also found application in the modeling of credit losses of loans under the default model with correlated defaults. The paper \cite{PenH} emphasizes the importance of modelling correlated defaults (in particular, when the correlations are negative) in credit risk modelling.

%
%

\section{Forward Continuation of the Backward Simulation for Mixed Poisson Processes}
\label{sec:fb_mpp}

Up to this point we have discussed the construction of MPPs within some interval $[0,T]$. A natural question to ask is, what if we want to simulate the process forwards in time, past the original simulation interval. One solution to this was introduced in the Poisson setting in \cite{akm}, known as the Forward Continuation (FC) to Backward Simulation (BS), where the joint distribution was preserved at various future time points, with the interval in between filled in by BS. The main idea of the FC method is to retain the independent increments property.
Note that since the notion of the linear correlation coefficient and our choice of 
multivariate extremal dependency concept is pairwise in nature, the discussion in 
the bivariate setting generalizes immediately to the multivariate case.

In the more general setting of MPPs, the conditional independence of the increments 
poses a challenge. We show that with the right construction, the arguments in 
the Poisson setting extend naturally to the MPP setting. Consider a sequence of time 
intervals $[0, T)$, $[T, 2T)$, $\dots, [mT, (m+1)T)$. Suppose that a bivariate MPP 
$(X_t,Y_t)$ has been simulated in $[0, T)$ using BS and that we wish to continue forward 
the process $(X_{T+\tau},Y_{T+\tau}) = (X_T+\Delta X_\tau, Y_T+\Delta Y_\tau)$ in 
$[T, 2T)$. At the end of the interval $[T, 2T)$, draw a new independent version of 
$(X_T,Y_T)$ and add it to the original $(X_T,Y_T)$ to obtain $(X_{2T},Y_{2T})$ at 
time $2T$, i.e, take $(X_{2T},Y_{2T}) \, \eqod\, (X_T,Y_T)$. For $0 \leq \tau < T$, 
the process $(X_{T+\tau},Y_{T+\tau})$ can be constructed by conditional uniformity given 
the number of events in $[T,2T)$, i.e by Backward Simulation. This retains the independence 
of the marginal increments since $X_T$ is independent of $\Delta X_\tau$, and similarly for $Y_T$, 
yet preserves the joint distribution of the bivariate process $(X_t,Y_t)$ even though, 
in the MPP setting, the marginal increments are only conditionally independent. 
Note that $X_{T+\t}$ remains a MPP due to our construction and the superposition property 
of MPPs \cite{Grand}.

\subsection{Forward Time Structure of Correlations}
\label{subsect:fwd_time_struct_corr}
Now that we have a method to extend a bivariate MPP defined initially within 
the interval $[0,T)$ to an interval $[mT,(m+1)T)$, it is natural to analyze 
the behavior of the correlation coefficient on each interval. It turns out that 
we can prove asymptotic stationarity of the  correlation coefficient 
under Forward Continuation. Indeed, the covariance of the processes $X_t$ and $Y_t$ 
at time $T+\t$ can be written as 
\begin{equation*}
	\cov(X_{T+\tau},Y_{T+\tau}) = \cov(X_T,Y_T) + \cov(X_\tau,Y_\tau)
\end{equation*}
since $X_T$ is independent of $\Delta Y_\tau$ and vice versa. From (\ref{eq:covariance_t_T}), we have:
\begin{equation}
	\cov(X_{T+\tau},Y_{T+\tau}) = (1+\frac{\tau^2}{T^2})\cov(X_T,Y_T).
	\label{eq:cov_fn_t}
\end{equation}
Dividing each side by $\sigma{(X_{T+\tau})}\,\sigma{(Y_{T+\tau})}$, we obtain the correlation coefficient
\begin{align}
	\rho(T+\tau) & = \rho(T)\,(1+\frac{\tau^2}{T^2})\cdot\frac{\sigma{(X_T)}\sigma{(Y_T)}}{\sigma{(X_{T+\tau})}\sigma{(Y_{T+\tau})}}. \nonumber 
\end{align}
Using a similar argument we can extend (\ref{eq:covariance_t_T}) to
\begin{equation}
	\cov(X_{mT+\tau},Y_{mT+\tau}) = (m+\frac{\tau^2}{T^2})\cov(X_T,Y_T).
\end{equation}
for $m=0,1,2,\dots$
Thus,
\begin{align}
	\rho(mT+\tau) & = \rho(T)(m + \frac{\t^2}{T^2})\cdot\frac{\sigma{(X_T)}\sigma{(Y_T)}}{\sqrt{\sigma^2{(X_{mT}+\Delta X_\t)}}\sqrt{\sigma^2{(Y_{mT}+\Delta Y_\t)}}} \nonumber \\
	& = \rho(T)(m + \frac{\t^2}{T^2})\cdot\frac{\sigma{(X_T)}\sigma{(Y_T)}}{\sqrt{\sigma^2{(X_{mT})+\sigma^2(\Delta X_\t)}}\sqrt{\sigma^2{(Y_{mT})+\sigma^2(\Delta Y_\t)}}} \nonumber \\
	& = \rho(T)(m + \frac{\t^2}{T^2})\nonumber \\
	& \qquad \qquad \boldsymbol{\cdot} \frac{\sigma{(X_T)}\sigma{(Y_T)}}{\sqrt{m\,\sigma^2{(X_{T})+\sigma^2(X_\t)}}\sqrt{m\,\sigma^2{(Y_{T})+\sigma^2(Y_\t)}}}. \label{eq:fwd_cont_form}
\end{align}
In going from the first line in (\ref{eq:fwd_cont_form}) to the second line, we use the property that the variance term 
$\sigma^2(X_{mT}+\Delta X_{\t})$ can be decomposed as $\sigma^2(X_{mT}) + \sigma^2(\Delta X_{\t})$  since $\Delta X_\t \, \eqod \, X_\t$ 
and $X_{mT}$ is independent of $ \Delta X_\t$. Similarly, 
$\sigma^2(Y_{mT} + \Delta Y_\t) = \sigma^2(Y_{mT}) + \sigma^2(\Delta Y_\t)$. 
In going from the second line in (\ref{eq:fwd_cont_form}) to the third line, we use the property that $\sigma^2(X_{mT})$ can be 
written as $m \,\sigma^2(X_{T})$ which can be seen as follows. Consider that $\sigma^2(X_{2T}) = \sigma^2(X_{T+T}) = 
\sigma^2(X_T+\hat{X}_T)   = \sigma^2(X_{T}) + \sigma^2(\hat{X}_{T}) = 2\sigma^2(X_T)$, where $X_T$ and $\hat{X}_T$ are iid. So, by induction on m, we get that $\sigma^2(X_{mT}) = 
m\sigma^2(X_T)$. Similarly, $\sigma^2(Y_{mT}) = m\sigma^2(Y_T)$.
\begin{theorem}[Asymptotic Stationarity of the Forward Continuation]
The correlation $\rho(mT+\t)$ achieves asymptotic stationarity as $m\to\infty$:
\begin{equation}
	\lim_{m \to \infty} \rho(mT+\t) = \rho(T), \quad \tau \in [0,T].
\end{equation}
\label{thm:asympt_statonarity}
\end{theorem}
\begin{proof}

The R.H.S. of (\ref{eq:fwd_cont_form}) can be rewritten as follows
\begin{align*}
\rho(T)(m + \frac{\t^2}{T^2})\cdot & \frac{\sigma{(X_T)}\sigma{(Y_T)}}{\sqrt{m\,\sigma^2{(X_{T})+\sigma^2(X_\t)}}\sqrt{m\,\sigma^2{(Y_{T})+\sigma^2(Y_\t)}}} \\
= \rho(T)(m + \frac{\t^2}{T^2})\cdot & \frac{1}{m} \cdot \frac{\sigma{(X_T)}\sigma{(Y_T)}}{\sqrt{\sigma^2{(X_{T})+(1/m)\sigma^2(X_\t)}}\sqrt{\sigma^2{(Y_{T})+(1/m)\sigma^2(Y_\t)}}}. \\
\end{align*}
Passing to the limit as $m \to \infty$ in the standard manner, we obtain that
\begin{equation*}
\lim_{m \to \infty} \rho(mT+\t) = \rho(T)
\end{equation*}
as was to be proved.
\end{proof}

A graphical illustration of Theorem \ref{thm:asympt_statonarity} can be found in Figure \ref{fig:fwdbkwd_cont_nb}, which shows the good agreement between the analytic and numerical results. Note that while the correlations at the calibrated integer grid points $nT$ are exact, the correlation structure in between grid points, generated via the FC method, requires a few time periods in order to settle to the asymptotic value of $\rho(T)$. 

The first few time periods can be used as ``burn in" periods to achieve a more stable and accurate desired value for the correlation between processes for the in-between time periods as filled in by BS. This can be also used to construct bivariate MPPs that display a constant time structure of correlation at some given correlation value instead of exhibiting the linear behavior in the interval $[0,T]$. Note that, depending on the choice of simulation, the process obtained by simulating purely through BS on [0,2T) will have a different correlation structure than a process simulated using BS on [0,T) and FC on [T,2T). This is a user-defined choice that is dependent on the needs of the application.

\begin{figure}[h]
\centering
\includegraphics[width=1\linewidth]{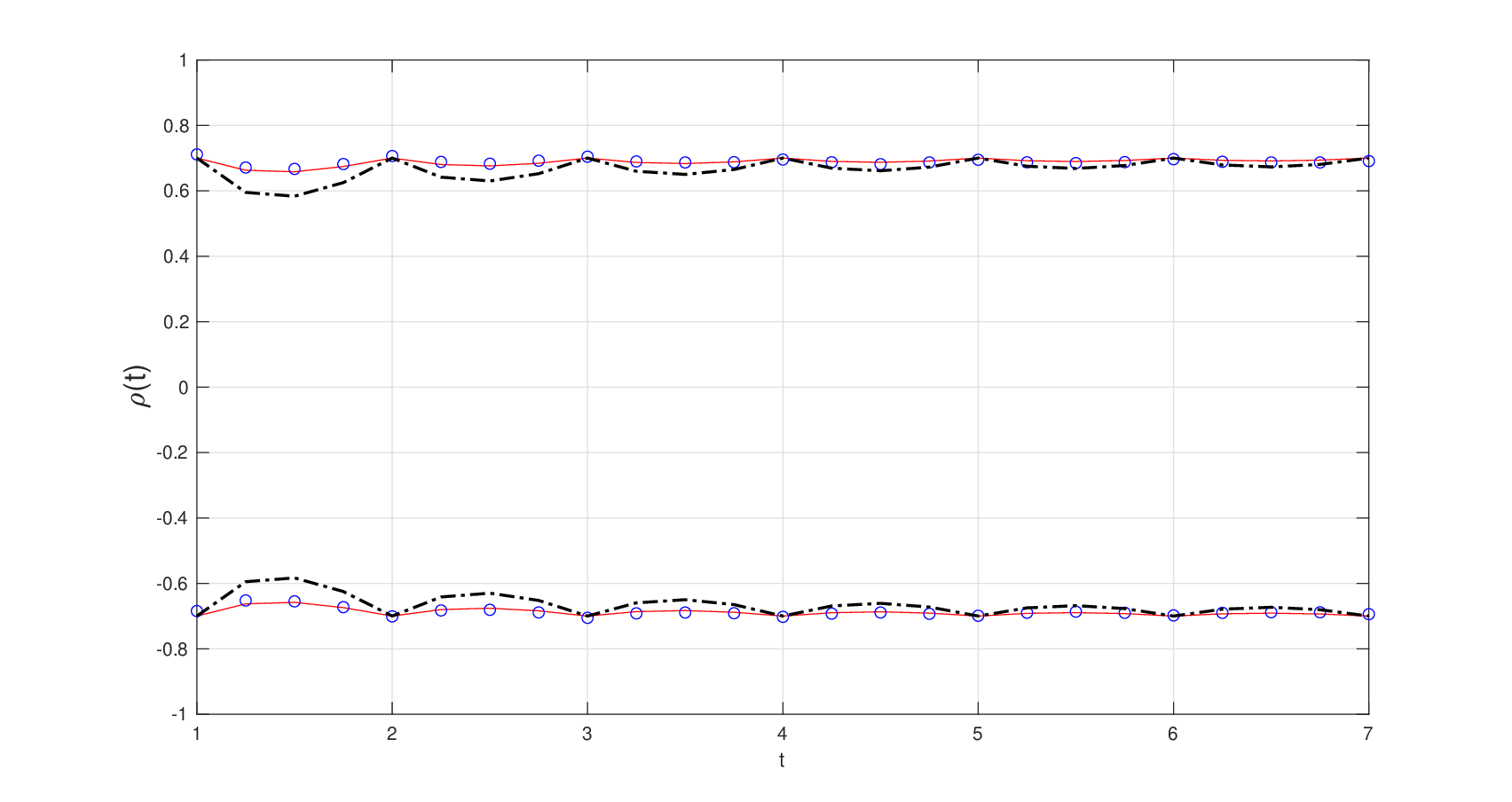}
\caption{Correlation structure for a bivariate NB process where the first NB process has a mean of 5 and a variance of 5 and the second NB process has a mean of 5 and a variance of 30. The bivariate process was calibrated to a positive correlation of $\rho(1) = 0.7$ and a negative correlation of $\rho(1) = -0.7$ at simulation time $T=1$. Then, the bivariate process was extended forward via Forward Continuation to $T=7$. The blue circles represent the correlations from Monte Carlo simulations of the bivariate NB process constructed through FC of the BS. The red line represents the correlation obtained analytically. The black dotted line represents the correlation structure of a bivariate Poisson process with the same mean parameters and calibrated to the same positive and negative correlations, derived analytically. }
\label{fig:fwdbkwd_cont_nb}
\end{figure}

%
%
\section{Concluding remarks}
\label{sec_cr}
In this paper, we extended the Backward Simulation (BS) method and the Forward Continuation of the BS method from the class of Poisson processes to the more general class of multivariate Mixed Poisson Processes. The advantages of the Backward approach over the Forward approach for generating sample paths of multivariate Mixed Poisson processes in some simulation interval $[0,T]$ are numerous: simple and efficient simulation in $d$-dimensions; specification of a dependency structure in the form of a given correlation matrix $C$ at terminal simulation time $T$; a wider range of possible correlations between the marginal distributions.

The Backward Simulation approach is applicable to any process that exhibits the order statistic property. 
For example, Backward Simulation is applicable to the class of Negative Binomial L\'{e}vy Processes \cite{bae2}. 
In fact, it is applicable to a more general class of processes known as Compound Poisson Processes which also posses 
a linear correlation structure under BS. It is also applicable to the inhomogeneous Poisson Processes. Extending the BS approach to Compound Poisson Processes and inhomogeneous Poisson Processes will be the 
subject of our forthcoming work.

\section*{Acknowledgment(s)}

This research was supported in part by the Natural Sciences and Engineering Research Council (NSERC) of Canada.





\printbibliography

@article{KAY,
author = {Duch, K. and Kreinin, A. and Jiang, Y.},
title = {New Approaches to Operational Risk Modeling},
journal = {IBM Journal of Research and Development},
volume = {3},
pages = {31-45},
year = {2014}
}

@incollection{Kre,
author = {Kreinin, A.},
title = {Correlated {P}oisson Processes and Their Applications in Financial Modeling},
editor = {Akansu, A. N and Kulkarni, S. R and Malioutov, D. M},
booktitle = {Financial Signal Processing and Machine Learning},
publisher = {John Wiley \& Sons},
pages = {191-230},
chapter = {9},
year = {2016}
}

@article{McNeil,
title={Common {P}oisson shock models: applications to insurance and credit risk modelling},
author={Lindskog, F. and McNeil, A. J.},
journal={ASTIN Bulletin: The Journal of the IAA},
volume={33},
number={2},
pages={209--238},
year={2003}
}

@article{akm,
title={Correlated Multivariate {P}oisson Processes and Extreme Measures},
author={Chiu, M. and Jackson, K. R. and Kreinin, A.},
journal={Model Assisted Statistics and Applications},
volume={12},
number={4},
pages={369--385},
year={2017},
publisher={IOS Press}
}

@article{Griff,
title={Aspects of correlation in bivariate {P}oisson distributions and processes},
author={Griffiths, R. C. and Milne, R. K. and Wood, R.},
journal={Australian \& New Zealand Journal of Statistics},
volume={21},
number={3},
pages={238--255},
year={1979},
publisher={Wiley Online Library}
}

@book{Grand,
title={Mixed Poisson Processes},
author={Grandell, J.},
year={1997},
publisher={CRC Press}
}

@article{Aue,
title={{LDA} at Work: {D}eutsche {B}ank's Approach to Quantifying Operational Risk},
journal={The Journal of Operational Risk},
author={Aue, F., and Kalkbrener, M.},
year={2006},
pages={49--95},
volume={1},
number={4}
}

@article{BNO,
title={Negative binomial processes},
author={Barndorff-Nielsen, O. and Yeo, G. F.},
journal={Journal of Applied Probability},
volume={6},
number={3},
pages={633--647},
year={1969},
publisher={Cambridge University Press}
}

@phdthesis{Lundb,
title={On Random Processes and their Application to Sickness and Accident Statistics},
author={Lundberg, O.},
year={1964},
school={Almqvist \& Wiksell, Uppsala.}
}

@article{EmbPuc,
title={Aggregating risk capital, with an application to operational risk},
author={Embrechts, P. and Puccetti, G.},
journal={The Geneva Risk and Insurance Review},
volume={31},
number={2},
pages={71--90},
year={2006},
publisher={Springer}
}

@book{RCont,
title={Financial Modelling with Jump Processes},
author={Cont, R. and Tankov, P.},
year={2004},
publisher={CRC Press}
}

@article{Bock2,
title={Multivariate models for operational risk},
author={B{\"o}cker, K. and Kl{\"u}ppelberg, C.},
journal={Quantitative Finance},
volume={10},
number={8},
pages={855--869},
year={2010},
publisher={Taylor \& Francis}
}

@article{BN1,
title={Non-{G}aussian {O}rnstein--{U}hlenbeck-based models and some of their uses in financial economics},
author={Barndorff-Nielsen, O. and Shephard, N.},
journal={Journal of the Royal Statistical Society: Series B (Statistical Methodology)},
volume={63},
number={2},
pages={167--241},
year={2001},
publisher={Wiley Online Library}
}

@article{Chav,
title={Quantitative models for operational risk: extremes, dependence and aggregation},
author={Chavez-Demoulin, V. and Embrechts, P. and Ne{\v{s}}lehov{\'a}, J.},
journal={Journal of Banking \& Finance},
volume={30},
number={10},
pages={2635--2658},
year={2006}
}

@book{Shev,
title={Modelling {O}perational {R}isk {U}sing {B}ayesian {I}nference},
author={Shevchenko, P.},
year={2011},
publisher={Springer Science \& Business Media}
}

@book{Panj,
title={Operational {R}isk: {M}odeling {A}nalytics},
author={Panjer, H. H.},
year={2006},
publisher={John Wiley \& Sons}
}

@article{Pet,
title={Dynamic operational risk: modeling dependence and combining different sources of information},
author={Peters, G. and Shevchenko, P. and Wuthrich, M.},
journal={Journal of Operational Risk},
volume={4},
number={2},
pages={69--104},
year={2009}
}

@phdthesis{Zocher,
title={Multivariate {M}ixed {P}oisson {P}rocesses},
author={Zocher, M.},
year={2005},
school={Almqvist \& Wiksell, Uppsala.}
}

@phdthesis{Chiu,
title={Correlated {M}ultivariate {P}oisson {P}rocesses},
author={Chiu, M.},
year={2020},
school={University of Toronto}
}

@misc{zoe,
title={A Modified Simplex Method for Solving ${A}x = b$, $x \ge 0$, for Very Large {A}
Arising from a Calibration Problem},
author={MacDonald, Z.},
year={2020},
school={University of Toronto},
howpublished = {    Computer Science Department,
                    University of Toronto},
url={http://www.cs.toronto.edu/NA/reports.html#Zoe_MacDonald_MSc_Research_Paper}
}

@article{Dian,
title={Dependent events and operational risk},
author={Powojowski, M. and Reynolds, D. and Tuenter, H.},
journal={Algo Research Quarterly},
volume={5},
number={2},
pages={65--73},
year={2002}
}

@article{Nesl,
title={Infinite mean models and the {LDA} for operational risk},
author={Ne{\v{s}}lehov{\'a}, J. and Embrechts, P. and Chavez-Demoulin, V.},
journal={Journal of Operational Risk},
volume={1},
number={1},
pages={3--25},
year={2006}
}

@article{Nels,
title={Discrete bivariate distributions with given marginals and correlation},
author={Nelsen, R.},
journal={Communications in Statistics-Simulation and Computation},
volume={16},
number={1},
pages={199--208},
year={1987}
}

@article{Hoeffd,
title={Masstabinvariante korrelations-theorie.},
author={Hoeffding, W.},
journal={Schriften Math. Inst. Univ. Berlin.},
volume={2},
pages={181-233},
year={1940},
}

@article{Whitt,
title={Bivariate distributions with given marginals},
author={Whitt, W.},
journal={The Annals of Statistics},
pages={1280--1289},
year={1976},
publisher={JSTOR}
}

@article{Frechet,
title={Sur les tableaux de corr{\'e}lation dont les marges sont donn{\'e}es},
author={Fr{\'e}chet, M.},
journal={Revue de L'Institut International De Statistique},
volume={14},
pages={53--77},
year={1951}
}

@book{Devroye,
title={Non-Uniform Random Variate Generation},
author={Devroye, L.},
year={1986},
publisher={Springer}
}

@article{baeKrenin,
author = {Bae, T. and Kreinin, A.},
title = {A backward construction and simulation of correlated {P}oisson processes},
journal = {Journal of Statistical Computation and Simulation},
volume = {87},
number = {8},
pages = {1593-1607},
year  = {2017},
publisher = {Taylor & Francis},
doi = {10.1080/00949655.2016.1277428},
URL = {https://doi.org/10.1080/00949655.2016.1277428},
eprint = {https://doi.org/10.1080/00949655.2016.1277428}
}

@article{bae2,
author = {Bae, T. and Mazjini, M.},
title = {Backward {S}imulation of {C}orrelated {N}egative {B}inomial {L}\'{e}vy {P}rocesses},
journal = {Mathematics and Statistics},
volume = {7},
number = {5},
pages = {191-196},
year  = {2019},
doi = {10.13189/ms.2019.070505},
}

@article{crump1975point,
  title={On point processes having an order statistic structure},
  author={Crump, K. S.},
  journal={Sankhy{\=a}: The Indian Journal of Statistics, Series A},
  pages={396--404},
  year={1975},
  publisher={JSTOR}
}

@book{rachev1998mass,
  title={Mass Transportation Problems: Volume I: Theory},
  author={Rachev, S. T. and R{\"u}schendorf, L.},
  year={1998},
  publisher={Springer Science \& Business Media}
}

@article{puccetti2015,
author = "Puccetti, G. and Wang, R.",
doi = "10.1214/15-STS525",
fjournal = "Statistical Science",
journal = "Statist. Sci.",
month = "11",
number = "4",
pages = "485--517",
publisher = "The Institute of Mathematical Statistics",
title = "Extremal {D}ependence {C}oncepts",
url = "https://doi.org/10.1214/15-STS525",
volume = "30",
year = "2015"
}

@book{feller2008,
title={An Introduction to Probability Theory and Its Applications},
author={Feller, Willliam},
volume={1},
year={2008},
publisher={John Wiley \& Sons}
}

@book{nocedal2006numerical,
title={Numerical Optimization},
author={Nocedal, Jorge and Wright, Stephen},
year={2006},
publisher={Springer Science \& Business Media}
}

\end{document}